\documentclass[prl,twocolumn,showpacs,superscriptaddress,amsmath,amssymb]{revtex4-1}
\usepackage[utf8]{inputenc}
\usepackage[draft]{fixme}
\usepackage[ruled]{algorithm}
\usepackage{algpseudocode}
\usepackage{graphicx,xspace}
\usepackage{soul}
\algnotext{EndFor}
\algnotext{EndIf}

\usepackage{natbib}
\usepackage{amsthm}
\newtheorem{lemma}{Lemma}
\newcommand{\ket}[1]{| #1 \rangle}
\newcommand{\bra}[1]{\langle #1|}

\newcommand{\cantor}{\{0,1\}^\omega}

\newcommand{\uph}{\upharpoonright}
\newcommand{\NN}{\mathbb{N}}

\newcommand{\VV}{\mathbf{V}}

\usepackage{color}

\usepackage{changes}

\usepackage[textwidth=0.5in,textsize=scriptsize]{todonotes}

\DeclareMathOperator{\Tr}{Tr}

\newtheorem{observation}{Observation}

\begin{document}
	\title{Distinguishing computable mixtures of quantum states}
	\author{Ignacio H. L\'opez Grande}
	\affiliation{DEILAP, CITEDEF-CONICET, Villa Martelli, Buenos Aires, Argentina}
	\author{Gabriel Senno}
	\affiliation{ICFO-Institut de Ciencies Fotoniques, The Barcelona Institute of Science and Technology, 08860 Castelldefels (Barcelona), Spain}
	\author{Gonzalo de la Torre}
	\affiliation{ICFO-Institut de Ciencies Fotoniques, The Barcelona Institute of Science and Technology, 08860 Castelldefels (Barcelona), Spain}
	\author{Miguel A. Larotonda}
	\affiliation{DEILAP, CITEDEF-CONICET, Villa Martelli, Buenos Aires, Argentina}
	\author{Ariel Bendersky}
	\affiliation{Universidad de Buenos Aires, Facultad de Ciencias Exactas y Naturales, Departamento de Computación, 1428 Buenos Aires, Argentina}
\affiliation{CONICET-Universidad de Buenos Aires, Instituto de Investigación en Ciencias de la Computación (ICC),
1428 Buenos Aires, Argentina}
	\author{Santiago Figueira}
\affiliation{Universidad de Buenos Aires, Facultad de Ciencias Exactas y Naturales, Departamento de Computación, 1428 Buenos Aires, Argentina}
\affiliation{CONICET-Universidad de Buenos Aires, Instituto de Investigación en Ciencias de la Computación (ICC),
1428 Buenos Aires, Argentina}	
\author{Antonio Ac\'\i n}
	\affiliation{ICFO-Institut de Ciencies Fotoniques, The Barcelona Institute of Science and Technology, 08860 Castelldefels (Barcelona), Spain}
	\affiliation{ICREA--Institucio Catalana de Recerca i
	Estudis Avan\c{c}ats, Lluis Companys 23, 08010 Barcelona, Spain}

	\begin{abstract}
		In this article we extend results from our previous work [Bendersky, de la Torre, Senno, Figueira and Acín, Phys. Rev. Lett. 116, 230406 (2016)] by providing a protocol to distinguish in finite time and with arbitrarily high success probability any algorithmic mixture of pure states from the maximally mixed state. Moreover, we introduce a proof-of-concept experiment consisting in a situation where two different random sequences of pure states are prepared; these sequences are indistinguishable according to quantum mechanics, but they become distinguishable when randomness is replaced with pseudorandomness within the preparation process.

	\end{abstract}
	
	
	\maketitle

	\section{Introduction}

With the advance of the experimental realization of quantum protocols, the most widely used class of setups consists of classical systems controlling quantum ones \cite{prevedel2007high,takeda2013deterministic,barrett2004deterministic,
takesue2005differential}. Being classical, the control systems are limited in the type of operations they can perform, and this has implications on what can be achieved by the setups they control.	
In particular, as it was shown in \cite{Bendersky2016}, if one intends to prepare a maximally mixed state by means of a computer pseudorandomly choosing pure states from a given basis, there is an algorithm that can distinguish such a preparation from an adequately prepared maximally mixed state, without any knowledge of the mixing procedure.

In this work we extend the ideas from \cite{Bendersky2016} in two ways. First, we generalize the theoretical result by showing that any preparation performed by a computer intended to generate the maximally mixed state, and not just those in which the states are chosen from a predefined basis, can be distinguished from the maximally mixed state. Second, we present a proof of concept experiment in which we distinguish two computable preparations that if carried out with randomness would be indistinguishable.
	
This article is organized as follows. First, we introduce the tools from the theory of algorithmic randomness which we will need later on. Second, we review the distinguishing protocol from \cite{Bendersky2016}. Third, we present its generalization to arbitrary computable preparations. Finally, we present results for a proof of concept experiment implementing a widely used scenario in which a pseudorandom function is used to pick pure states from a given basis.
	
	\section{Preliminaries}

Central to the distinguishing protocols we will describre in the following sections is the idea of an \emph{algorithmically random} sequence of symbols. Roughly, an infinite sequence of symbols from some finite alphabet $\Sigma$ is random in an algorithmic sense, if it lacks any regularity detectable by effective means. Randomness tests, also called \emph{Martin-L\"of tests (ML-tests)} \cite{MARTINLOF1966602} , are defined to detect some specific regularity. This `detection' of non-random sequences must be computably approximable, with incrementing levels of accuracy or significance. A test is a collection of sets $V_m$ of possible prefixes of sequences that do not look random. As we increase $m$, the identification of non-randomness gets more and more fine-grained, leaving in the limit a null measure set of non-random sequences. The \emph{Martin-L\"of random (ML-random)} sequences are those not detectable by any possible ML-test.

Formally, let $\Sigma^*$ be the set of all finite strings with symbols from $\Sigma$. A Martin-L\"of test is a sequence $(V_m)_{m\in\NN}$ of sets $V_m\subseteq\Sigma^*$ with two  properties:

\begin{enumerate}
\item {\em Effectiveness.} There is a Turing machine that given $m$ and $i$, produces the $i$-th string of $V_m$ (notice that in general there are infinitely many strings in $V_m$). It is not possible to computably determine if a string {\em is not} in $V_m$, but we can computably enumerate all strings that are in.\label{prop:effectiveness}

\item {\em Null class.} Let $\lambda$ be the uniform measure on the space $\Sigma^\omega$ of infinite sequences with symbols from $\Sigma$ and, for $A\subseteq\Sigma^*$, let $[A]\subseteq\Sigma^\omega$ denote the set of sequences with prefixes in $A$. Then, we require each ML-test $(V_m)_{m\in\NN}$ to satisfy $\lambda [V_m]\leq |\Sigma|^{-m}$.
\end{enumerate}

We say that a sequence $Y\in\Sigma^\omega$ 
 is {\em ML-random} if no ML-test $(V_m)_{m\in\NN}$ can capture $Y$ in {\em all} its levels of accuracy, that is if for no test $(V_m)_{m\in\NN}$ we have $Y\in\bigcap_m[V_m]$. Informally, if $Y\in [V_m]$ then we reject the hypothesis that $Y$ is random with significance level $|\Sigma|^{-m}$.

One of the most important features of the theory of Martin-L\"of randomness is the existence of a universal ML-test, i.e. a test $(U_m)_{m\in\NN}$ such that a sequence $Y\in\Sigma^{\omega}$ is ML-random iff $Y\not\in\bigcap_m[U_m]$. Since $\lambda\bigcap_m[U_m]=0$, this implies that the set of ML-random sequences has measure 1. In other words, the sequence of independent throws of a $|\Sigma|$-faced dice is ML-random with probability~1.

Let $Y\upharpoonright n$ denote the prefix of length $n$ of the sequence~$Y$. Observe that, although $U_m=\{s_1,s_2,\dots\}$ will, in general, be infinite, if $Y\in[U_m]$ then for large enough $n$ we have
that all the infinite sequences extending $Y\upharpoonright n$ belong to $[\{s_1,\dots,s_n\}]$. This last expression can be seen as the $n$-th approximation of $[\{s_1,s_2,\dots\}]$. Hence if $Y\in\bigcap_m[U_m]$, then for every $m$ there is $n$ such that any extension of $Y\upharpoonright n$ is included in the $n$-th approximation of $[U_m]$.

Intuitively, we expect a random sequence $Y\in\Sigma^\omega$ to satisfy the law of large numbers,
\begin{align}\label{lawlargenumb}
\lim_n \frac{|\{i<n~|~Y(i)=b\}|}{n}=\frac{1}{|\Sigma|}\mbox{ for all }b\in\Sigma.
\end{align}
Furthermore, it is natural to ask of random sequences that there be no algorithmic way of \emph{selecting} some subsequence
of it not satisfying \eqref{lawlargenumb} (say, for instance, a subsequence of all $0$s in the binary case). This property, known
as \emph{Church stochasticity} \cite{church1940concept}, is satisfied by ML-random sequences (see, e.g. \cite[Section 2.5.]{LVBook}) and we will use this fact in what follows.

	\section{Distinguishing pseudomixtures of quantum states}

In~\cite{Bendersky2016} we considered a scenario with two players, Alice and Bob, in which, first, Alice fixes a qubit basis, either the $\sigma_z$ basis or the $\sigma_x$ basis, and then, upon Bob's successive requests, pseudorandomly picks an eigenstate from the chosen basis and sends it to him. We gave a protocol for Bob to
distinguish the (initially unknown to him) preparation basis in finite time and with arbitrarily high
success probability. This implies that it is incorrect to characterize Bob's lack of
knowledge about the preparation basis with the maximally mixed state as one would do if Alice were using randomness.

The protocol followed by Bob has two steps. First, he alternatively measures the qubits being sent by Alice in the $\sigma_x$ and $\sigma_z$ basis. This generates two binary sequences: $X$ and $Z$ (see Fig. \ref{fig:alternation} for a schematic description). When he measures in the preparation basis, the corresponding sequence will be a subsequence (either the odd or the even positions) of the pseudorandom sequence being used by Alice; when he measures in the other basis, the resulting bits are, according to quantum mechanics, independent flips of a fair coin and, therefore, they give rise to a ML-random sequence with probability $1$. In the second step of the
protocol, Bob uses a universal ML-test $(U_m)_{m\in\NN}$ to distinguish between these two kind of sequences and hence find out the preparation basis. Namely, given a desired probability of error $\epsilon$, he computes $m:=\min_k [2^{-k}\leq\epsilon]$ and starts enumerating all the strings in $U_m=\{s_1,s_2,\dots\}$ until he finds some $n$ such that for $Y=Z$ or $Y=X$ it happens that
$$
[Y\upharpoonright n]\subseteq \bigcup_{i\leq n}[s_i],
$$
after which he claims that the box producing $Y$ is the one with the computer.
Since either $X$ or $Z$ is computable, and hence not ML-random, the last condition has to be satisfied for sufficiently large $n$. His claim is wrong when the ML-random sequence was captured by $[U_m]$ before the computable one was (of course, for some $m'>m$ the ML-random sequence would be out of $[U_{m'}]$). Hence, the probability of making this error is at most the probability for the coin flipping sequence to be inside $[U_m]$, and this is at most $2^{-m}\leq\epsilon$.

\begin{figure}[ht]
\includegraphics[width=0.85\columnwidth]{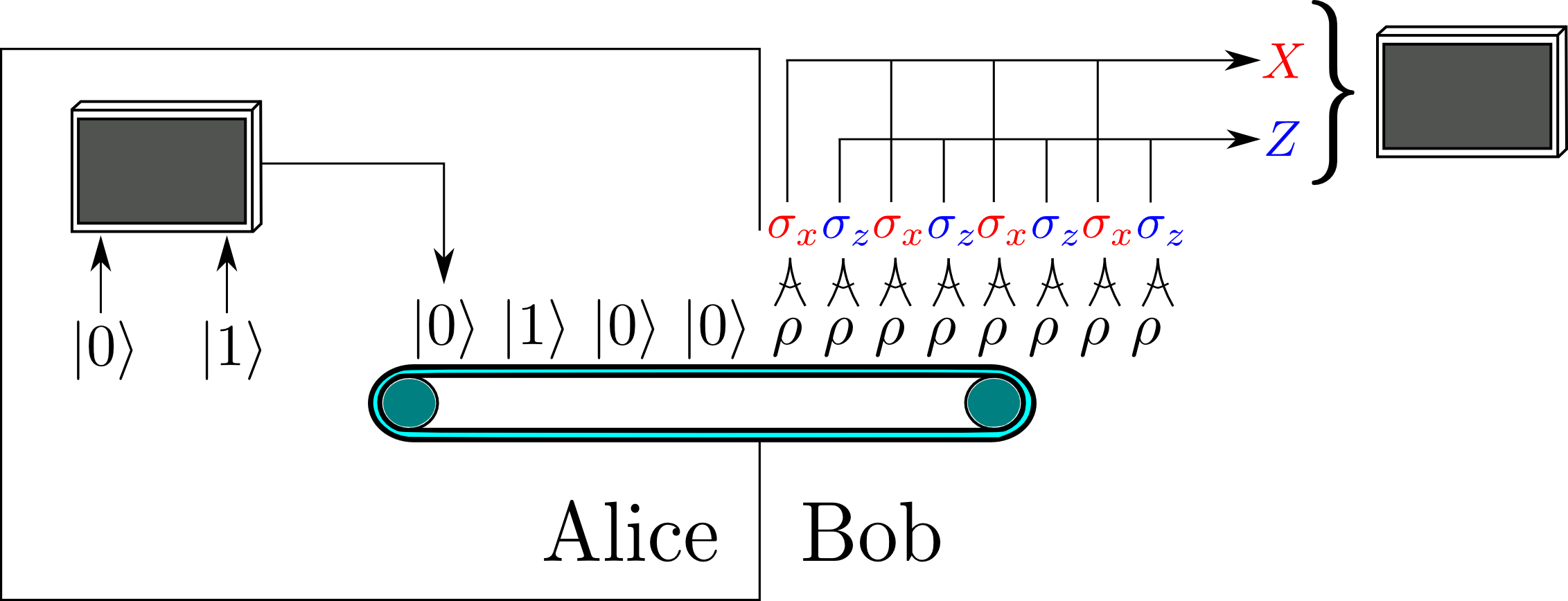}
\caption{Schematic description of the protocol given in~\cite{Bendersky2016} allowing a player Bob to tell if he is being given pseudrandom eigenstates of the $\sigma_x$ basis or of the $\sigma_z$ basis. \label{fig:alternation}}
\end{figure}

	\section{Generalized distinguishing protocol}

    In this Section we extend the results from~\cite{Bendersky2016}. We will consider a scenario in which there are two boxes providing qudits to an observer named Bob. One of the boxes prepares single qudit maximally mixed states (for instance, by preparing the maximally entangled bipartite state $\frac{1}{\sqrt{d}}\sum_i\ket{ii}$ and keeping one half while providing the observer with the other). The other box contains a computer producing, at each round $j$, $2d$ rational numbers \footnote{We work with rational numbers for simplicity, but any other computable ordered field (i.e. an ordered field whose field operations $(+,\cdot)$ and order relation $\leq$ are computable) would do.} $\{(r_k^{(j)},\phi_k^{(j)})\}_{k\leq d}$ with $\sum_k r_k^2=1$ and preparing a qudit in the state $\ket{\psi_j} := \sum_k r_k^{(j)} e^{i\phi_k^{(j)}}\ket{i}$. Bob, without any knowledge about which box is which, will face the problem of determining the one preparing the maximally mixed state. Our main result is a protocol for Bob to win this game with arbitrarily high probability and independently of the program being run by the computer.


    Before going to Bob's protocol, let us first note that if we fix a basis $\mathcal{B}$ and only allow the computer to pick eigenstates from such basis, a slight modification of the protocol from~\cite{Bendersky2016} allows Bob to distinguish between the boxes. Namely, if instead of alternating between measuring $\sigma_x$ and measuring $\sigma_z$ as in~\cite{Bendersky2016}, Bob measures the outputs of both boxes in the $\mathcal{B}$ basis, the $d$-ary sequence associated with the box which has the computer will be computable and the other, according to quantum mechanics, independent tosses of a fair coin and so Martin-L\"of random. Hence, our previous result applies. The situation we want to consider in this work is when there is no fixed preparation basis.
	
	Bob's protocol works as follows. In each round, he will perform an \emph{informationally complete} POVM $\{E_i\}_{i\leq N_d}$ (i.e. one for which any set of outcome probabilities specifies a unique state) to the qudits coming out of each of the boxes satisfying
\begin{align}\label{eqn:unbiasedness}
  \Tr(E_i\frac{\mathbb{I}}{d})=\frac{1}{N_d}\mbox{ for all }E_i.
\end{align}

It is easy to see that such POVMs exist in every dimension $d$ (see Appendix 1). This will give rise to
two $N_d$-ary sequences $B_1$ and $B_2$ formed by the results of the measurements over the qudits coming from boxes $1$ and $2$. Note at this point that although Bob measures finitely many times, the sequences are infinite in the sense that he can keep requesting qudits from both boxes and making as many measurements as he needs. 	As we will see now, sequences $B_1$ and $B_2$ have a distinctive feature that will allow Bob to distinguish which is the maximally mixed state and which is the one being produced by a computer.
	
	Let $r\in \left\lbrace 1,2 \right\rbrace$ be the box preparing the maximally mixed state and $c=3-r$ be the box with the computer inside. It follows from \eqref{eqn:unbiasedness} that, with probability $1$, the sequence $B_r$ will be Martin-L\"of random.
%
%
%
%
		On the other hand, with probability $1$, sequence $B_c$ will not be Martin-L\"of random. This is not straightforward, and we prove it next.

First, notice that from the fact that the POVM $\{E_i\}_{i\leq N_d}$ satisfies \eqref{eqn:unbiasedness} and it is informationally complete, it follows that
	
	\begin{observation}\label{obs:biasmeasurements}
		Let $\ket{\psi_j}$ be the pure state produced by box $c$ at round $j$. There is, at least, one $E_i$ such that
$\Tr(E_i\ket{\psi_j}\bra{\psi_j})>1/N_d$.
	\end{observation}
	
    This, together with the following lemma, will allow us to show that any computable preparation made by Alice is distinguishable from the correctly prepared maximally mixed state.

    \begin{lemma}\label{lemma:seqcompnotrandom}
    	With probability $1$, sequence $B_c$ is not ML-random.
    \end{lemma}

    \begin{proof}
	Following Observation \ref{obs:biasmeasurements}, without loss of generality, we assume that $E_k$ is such that
\begin{align}\label{assumption}
  \Tr(E_k\ket{\psi_n}\bra{\psi_n})> 1/N_d\mbox{ for infinitely many }n.
\end{align}

This means that there is an algorithmic way to identify a subsequence of $B_c$ not satisfying the law of large numbers (with probability $1$). Namely, let $h:\NN\to\NN$ be defined as
\begin{align*}
h(0)&:=0\\
h(n+1)&:=\min_{m} \left[[\Tr(E_k\ket{\psi_m}\bra{\psi_m})>\frac{1}{N_d}]\land [m>h(n)]\right]
\end{align*}
By assumption \eqref{assumption}, $h(n)$ is defined for all $n$. Next, by definition of $h$, with probability $1$ the sequence
	$$Y=B_c(h(0))B_c(h(1))B_c(h(2))\dots\in\{1,\dots,N_d\}^\omega,$$
	which is a subsequence of $B_c$, does not satisfy the law of large numbers \eqref{lawlargenumb}. Hence, noting that $\ket{\psi_m}$ is computable from $m$ (e.g. with Alice's program) and so $h$ is a computable function, we have that, with probability $1$, $B_c$ is not Church stochastic and so it is also not ML-random.
	\end{proof}


    We have proven that $B_c$ is not ML-random but $B_r$ is. Now the argument carries on as in \cite{Bendersky2016}. Namely, given a desired probability of error $\epsilon$, Bob computes $m:=\min_k [2^{-k}\leq\epsilon]$ and starts enumerating all the strings in $U_m=\{s_1,s_2,\dots\}$ until he finds some $n$ such that
$
[B_i\upharpoonright n]\subseteq \bigcup_{i\leq n}[s_i]
$
for some $i\in\{1,2\}$ and claims that box $i$ is the one with the computer.
Since, with probability 1, either $B_1$ or $B_2$ is not Martin-L\"of random, the last condition has to be satisfied for sufficiently large $n$ with probability 1. His claim is incorrect when the sequence ML-random was captured by $[U_m]$ which happens with probability $2^{-m}\leq\epsilon$.

    \section{Experimental test}
%
%
%
%
%
%
%
%
%

In this section we present a proof-of-concept realisation of the distinguishing protocol presented in \citep{Bendersky2016} and resumed above. In the next lines we describe the additions/modifications made to the theoretical scenario, arising from experimental considerations.

First, to account for experimental imperfections, we will work under the assumption of a noise model consisting of a flip probability $f$ in the observed symbols. That is, we consider the situation in which those results obtained when measuring the qubit states in the actual basis used by Alice are correct with probability $1-f$ (this simple noise has no effect on the results of measurements performed in the wrong basis).
For the sake of concreteness, we describe next an explicit algorithm for Bob to distinguish which of sequences of measurement outputs $X$ and $Z$ is the one corresponding to measuring in the preparation basis (see Fig. \ref{fig:alternation}). This algorithm, although less resistant to noise than the general protocol using ML-tests given in \cite{Bendersky2016}, is robust enough for the noise model we are considering.

Bob will dovetail between program number and the maximum time steps required for the simulation of this program on a (fixed) universal Turing machine $\VV$ (that is, he will simulate program 1 for 1 time step, then programs 1 and 2 for 2 time steps and so on). This is a common technique in computability theory to avoid non-halting programs (see e.g. \cite{davis1994computability}). For each program $p$ of length $|p|$ he will compute the Hamming distance (i.e. the number of different bits) between its output at time $t$ and the first $k|p|$ bits of the sequences $X$ and $Z$ (notated $X\uph k|p|$ and $Z\uph k|p|$ respectively). The parameter $k\in\NN$ will depend on the probability of success we are looking for. Whenever he finds a match for the first $k|p|$ bits, he halts and claims that the corresponding sequence is the computable one. 
Letting $q \in \mathbb{Q}$ be the fraction of bit flips in the prefixes, the pseudocode is Algorithm \ref{alg:distinguishNoise} below, where $d_H$ denotes Hamming distance.

\begin{algorithm}[H]\caption{The noise tolerant distinguishing protocol}\label{alg:distinguishNoise}
\label{alg:the_alg}
\begin{algorithmic}
\Require $q\in \mathbb{Q}$, $k \in \NN$ and $X,Z\in\cantor$, one of them being computable
\Ensure `$X$' or `$Z$' as the candidate for being computable; wrong answer with probability bounded by $O(2^{-k})$
\For{$t=0,1,2\dots$}
    \For{$p=0,\dots,t$}
     \If {$d_H(\VV_t(p),X\uph k|p|)<qk|p|$} \State output `$X$' and halt \EndIf
     \If {$d_H(\VV_t(p),Z\uph k|p|)<qk|p|$} \State output `$Z$' and halt \EndIf
    \EndFor
\EndFor
\end{algorithmic}
\end{algorithm}

In the appendix we show that the probability of error, i.e. the probability of Bob making a wrong claim about which of the two sequences $X$ and $Z$ is a subsequence (with its bits flipped with probability $q$) of Alice's sequence, is
\begin{equation}
 P_{{\rm err}}<\frac{2^{1+qk-k}\left(\frac{e}{q}\right)^{qk}}{1-2^{1+qk-k}\left(\frac{e}{q}\right)^{qk}}.
\end{equation}
and, it can be shown numerically that for $q\lesssim 0.21$ it goes to zero exponentially with $k$. This distinguishing protocol appeared in a first version \cite{bendersky2014implications} of the results published in \citep{Bendersky2016}.

Notice that Algorithm \ref{alg:distinguishNoise} --as it was the case with the protocol using a universal ML-test-- is independent of Alice's algorithm. This independence, however, comes at the expense of unfeasibility, because it is achieved through a search over the whole space of all Turing machines. Hence, the second implentation decision we make is to restrict the possible algorithms used by Alice to the $rand()$ function of Matlab using the Mersenne Twister default generator algorithm \cite{matsumoto1998mersenne}\} with initial seeds of a fixed maximum length $\ell_{max}$. In spite of being a simplified scenario, this still represents a quite usual experimental situation.
%
%
Finally, some minor changes to Algorithm \ref{alg:distinguishNoise} were required due to the non-deterministic nature of the emission and detection of \emph{Poissonian} single photon states used as physical implementation for qubits. The adapted protocol can be specifically stated as follows:
\begin{itemize}

\item Alice and Bob set the value of two parameters from the protocol: $\ell_{max}$ which determines the maximum length of the \textit{rand()} function seed to be used and $k$ which bounds to N = $k \times \ell_{max}$, the number of qubits to be transmitted on any run of the experiment.

\item Alice pseudo-randomly chooses one integer between $0$ and $2^{\ell_{max}}\text{-}1$ which is used as the initial value, or seed for the \emph{rand}() function. The output of  \emph{rand}() is binarized using the \emph{round}() function resulting on a string of $N$ \textit{pseudo-random} bits.


\item Alice chooses randomly (with fair coin randomness as explained below) the basis in which she will encode and send the string.

\item Alice sends the $N$ qubits to Bob. She encodes the binary string information in the photon polarization degree of freedom of a faint pulsed light beam.

\item Bob measures the  $\frac{N}{2}$ even and  $\frac{N}{2}$ odd elements, each in one of the mutual unbiased bases.


\item Bob, after measurement, computes the Hamming distance (for even and odd bits) between experimental data and the output of \textit{rand}() function with the different seeds. When the minimum Hamming distance condition is fulfilled Bob ends the search.

\item Finally Bob compares the state preparation ($\sigma_{x}$ or $\sigma_{z}$ mixtures) predicted by him with the mixture that was actually prepared by Alice to estimate the error probability ($P_{{\rm err}}$) of the prediction.

\end{itemize}


A complete experiment consists in several repetitions of the protocol sketched above. Every execution is divided in two parts;  the \emph{transmission} of qubits from Alice to Bob, followed by a \emph{search} routine, where Bob compares both bit strings with the strings generated by the \emph{rand}() function over all seeds of length bounded by $\ell_{max}$ as it is stated in the theoretical protocol. When Bob finds a string that resembles the experimental series up to a certain $d_H$ value, 
the search ends. The result is compared with the actual basis used by Alice and the wrong guesses are registered as errors. After this they repeat the procedure  with a new seed pseudo-randomly picked, and a new random emission basis choice. The bound for $d_H$ allows us to control the tolerance of the experiment against the Quantum Bit Error Rate (QBER).

One thing to be noticed is that Bob may not find a series that fulfills the desired Hamming distance condition. 
 This is a situation that is not present in the theoretical protocol. In this way every time that Bob doesn't find a match we compute the experiment as inconclusive and it is discarded. To overcome this issue, the parameters of the protocol (such as maximum $d_H$ allowed)  were set to guarantee that the probability of error occurrence was always greater than the probability of not finding any bit string fulfilling the condition. Under such assumptions, and using reasonable tolerances, we find that the ratio of inconclusive experiments to total number of errors was negligible.

The experiment involved 3100 repetitions of the \textit{transmission} and \textit{search} protocols. The total number of qubits transmitted on each repetition was fixed, and set by $k_{max} \times \ell_{max}$ (in this implementation $\ell_{max}=10$). The parameter $k$ determines the theoretical error probability for a given tolerance ($q$) and was set to take values between $1$ and $16$. This bounds the maximum number of compared bits on each Hamming distance calculation to $N=320$ ($\ell_{max}\times k_{max}$ bits for even and odd bits); that is the number of qubits that Alice sends to Bob on each run.

After the qubit transmission is finished, Bob begins the search procedure building a list of programs with the restricted family of seeds in the following way:

\begin{itemize}
\item the seeds $0$ and $1$ are assigned to the $1$-bit programs $0$ and $1$ respectively.
\item the seeds $0, 1, 2$ and $3$ are assigned to the 2-bit programs $00, 01, 10$ and $11$ respectively.
\item (...)
\item the seeds $0,1,\dots,2^{\ell_{max}}$ are assigned to the $\ell$-bit programs $000\dots 0,\dots,111\dots 1$ respectively.
\end{itemize}

Note that the resulting list has $2^{\ell_{max}+1}-2$ elements. Some programs appear repeated (e.g. the rows with bold letter in the table~\ref{table:tb_0} correspond to the seed 0 that appears 10 times in a list with $\ell_{max}=10$), with different associated lengths. This particular condition is required by the noise tolerant version of the protocol, as it depends on a known fact on computable sequences: every computable sequence can be generated by infinitely many different programs.

\begin{table}
\begin{tabular}{|c|c|c|c|}
  \hline
  \textbf{program \#} & \textbf{seed (bits) }& \textbf{pr. length} & \textbf{bit string}\\
  \hline
  \textbf{0} & \textbf{0} & \textbf{1} & \textbf{00111011...}\\
  \hline
  1 & 1 & 1 & 10000100...\\
  \hline
  \textbf{2} & \textbf{00} & \textbf{2} &\textbf{00111011...}\\
  \hline
  3 & 01 & 2 &10000100...\\
  \hline
  4 & 10 & 2 &10000100...\\
  \hline
  5 & 11 & 2 &01101000...\\
  \hline
  \textbf{6} & \textbf{000} & \textbf{3} &\textbf{00111011...}\\
  \hline
  (...) & (...) & (...)&(...)\\
  \hline
  $2^{\ell_{max}+1}-1$ & 1111111111... &$\ell_{max}$ & 10001111... \\
  \hline
\end{tabular}
\caption{\footnotesize{\emph{Search list used by Bob to compare the experimental data with the pseudorandom strings generated by the \emph{rand}}() \emph{function with different seeds}}\label{table:tb_0}}
\end{table}

Bob compares the first measured bit string with each row on the list, and stops the search when either the even or odd bits of the compared strings fulfil the Hamming Distance criterion. Finally he compares the basis for the mixed state preparation predicted by this protocol with the one that Alice actually used, for the error probability estimation.

\subsection{Experimental setup}

The above protocol was tested on a photonic setup, based on a modified BB84 Quantum Key Distribution (QKD) implementation \cite{lopezgrande2016autonomous} which consists of an emission stage that is able to send binary states coded in two different unbiased bases of the photon polarization, which are called computational basis and diagonal basis, and a reception stage for the quantum channel. Additionally, a classical communication channel is added for synchronization, transmission and data validation.

The four polarization qubits are obtained using attenuated coherent pulses generated with four infrared LEDs, controlled by a fast pulsed driver (optical pulse duration $25ns$ FWHM). Faint coherent pulses can be used as probabilistic single photon sources: on each pulse the photon number distribution is Poissonian. Unlike the theoretical protocol, where each qubit is sent and received deterministically, here the transmission of a qubit is probabilistic. As opposed to QKD, in this demonstration the fact that most of the emitted pulses have zero photons requires Alice to send each state several times until Bob makes a successful detection.

Polarization states are obtained by combining all the outputs from the LEDs in a single optical path using polarization beam splitters (PBS), a half waveplate retarder and a beamsplitter (BS).
A bandpass filter centered at $810$ nm narrows the photons bandwidth down to $10$ nm FWHM. A TTL clock pulse is sent to Bob every time a pulse is emitted in order to synchronize the optical pulses with the gated detection scheme.

\begin{figure}[h!]
    \centering
    \includegraphics[width=0.45\textwidth]{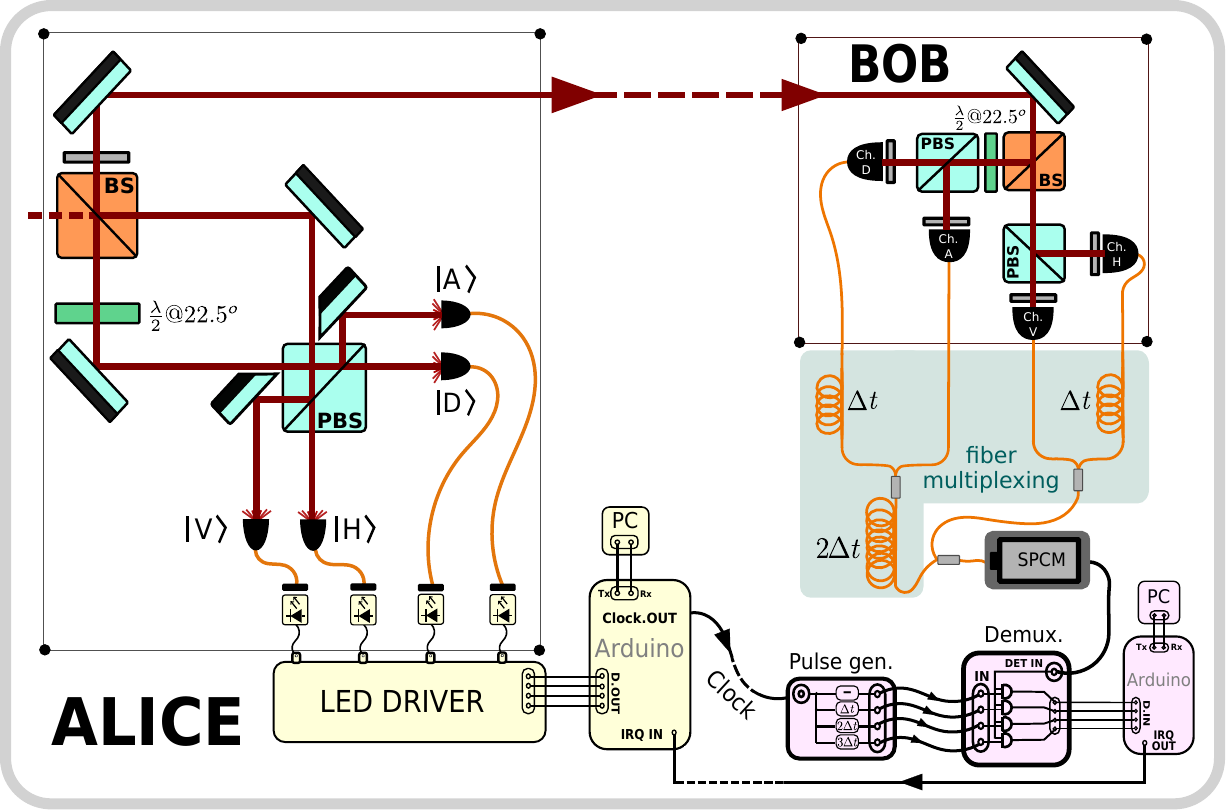}
  \caption{\footnotesize{\emph{Complete setup for implementing the transmission and search protocol: Qubits encoded in polarized faint pulses are produced by infrared LEDs. Light is coupled into and de-coupled from multimode fibers to obtain uniform beams for the four sources. The polarization state preparation is achieved by passing through a PBS (for H and V states) and an extra halfwave-plate for the D and A paths. A non-polarizing Beam Splitter cube couple the optical paths into an only exit light path. At the receiver's side a BS passively and randomly selects the detection basis for each incoming pulse. The outputs are coupled into multimode optical fibers, where different delays are imposed to make a polarization to time-bin transformation into a common output fiber. Finally a photon counter module and a temporal mask demultiplexer are used for detection.}}}
  \label{fig:0}
\end{figure}

At Bob's side the detection basis is passively and randomly selected with a BS. Each detection basis consists in a PBS with both outputs coupled into multimode fibers. An additional half-waveplate before one of the PBS allows for detection in the diagonal basis. We implement a polarization to time-bin transformation by adding different delays to each channel. This allows us to utilize a fiber multiplexing scheme with only one single photon detector  (figure \ref{fig:0}). Temporal masks generated using the clock pulse emitted by Alice act as demultiplexer and detection gating.

Programmable Arduino Mega 2560 boards are used to carry out the synchronization, communication and data processing tasks, for which specific interfacing peripherals were developed. A desktop personal computer generates the binary strings of pseudorandom bits using the Matlab function \emph{rand}(), and stores the bit strings. Finally a Quantum Random Number Generator (QRNG) based on which-path detections of single photons exiting a beam splitter is used for the realization of a random selection of the emission basis on each repetition of the experiment.

Alice sends each bit of the string repeatedly at a frequency of $170$ kHz until she receives an interruption signal, indicating that the qubit was correctly detected by Bob. Due to the probabilistic nature of the qubit transmission process each state may be sent several times before Bob makes a successful detection. In particular, given that the photon number distribution per pulse is Poissonian (with a mean photon value at the detector of $0.1$), on average one every ten pulses is detected. Furthermore, the detection base is randomly selected so $50\%$ of the detected photons are discarded by base mismatch. This results in an overall qubit transmission rate of $\frac{1}{20}$ per emitted pulse. 


\subsection{Complete Results and Simulations}
\label{sec:simul}

Herein we analyze the experimental results. We compare the performance of Bob at guessing the emission basis, with the error probability $P_{{\rm err}}$ obtained in \cite{bendersky2014implications}, and we also present additional data analysis aiming to explain the behavior of the error rate obtained.\\


As a result of each run, Bob gets two 160-bit length strings. $M_e$ are the outcomes of even qubits, measured in the computational basis and $M_o$ are the outcomes of odd qubits, measured in the diagonal basis. These two strings correspond to $Z$ and $X$ introduced in Algorithm \ref{alg:the_alg}. 
Bob compares these strings with the pair of strings from the program list $p_{e}$ and $p_{o}$, where the $p$ stands for the number of program evaluated.
Note that when evaluating a program of length $\ell $ just the first $k \times \ell $ bits of the transmitted string are taken into account to compute $d_H$. The whole $160$ bit string is only used in the Hamming distance measure of programs with $\ell_{max}=10$.

We calculate the Hamming distance between the strings, $d_H(p_{e},M_e\uph k \, \ell )$ and  $d_H(p_{o},M_e\uph k\,\ell )$, and the search finishes when one of them fulfils the tolerance criteria: $d_H(p_{i},M_i\uph k\,\ell )\leq \lfloor q \times k \times \ell  \rfloor$ from the noise tolerant protocol. In this experiment the tolerance parameter is set to $q=0.15$. The result of the search for each run is registered for a further estimation of the error rate $R_{{\rm err}}(k,q)$.

The probability of error in Bob's guess of the emission basis can be estimated for different values of the parameter $k$. Figure \ref{fig:exp&sim} shows the error rate obtained from the experimental data and from a computational simulation of the experiment, together with the theoretical bounds for the distinguishing -- noiseless and noise tolerant -- protocols.

The error rate as a function of $k$ remains always below the noise-tolerant limit and also above the noiseless theoretical bound (excluding the scenarios with values $k=1$ and $k=2$). The error shows some unexpected increments  for $k=7$ and $k=14$. This behavior arises due to the discrete nature of the number of errors allowed on each string comparison, and it is explained below.\\

\begin{figure}[ht]
    \centering
    \includegraphics[width=8.75cm]{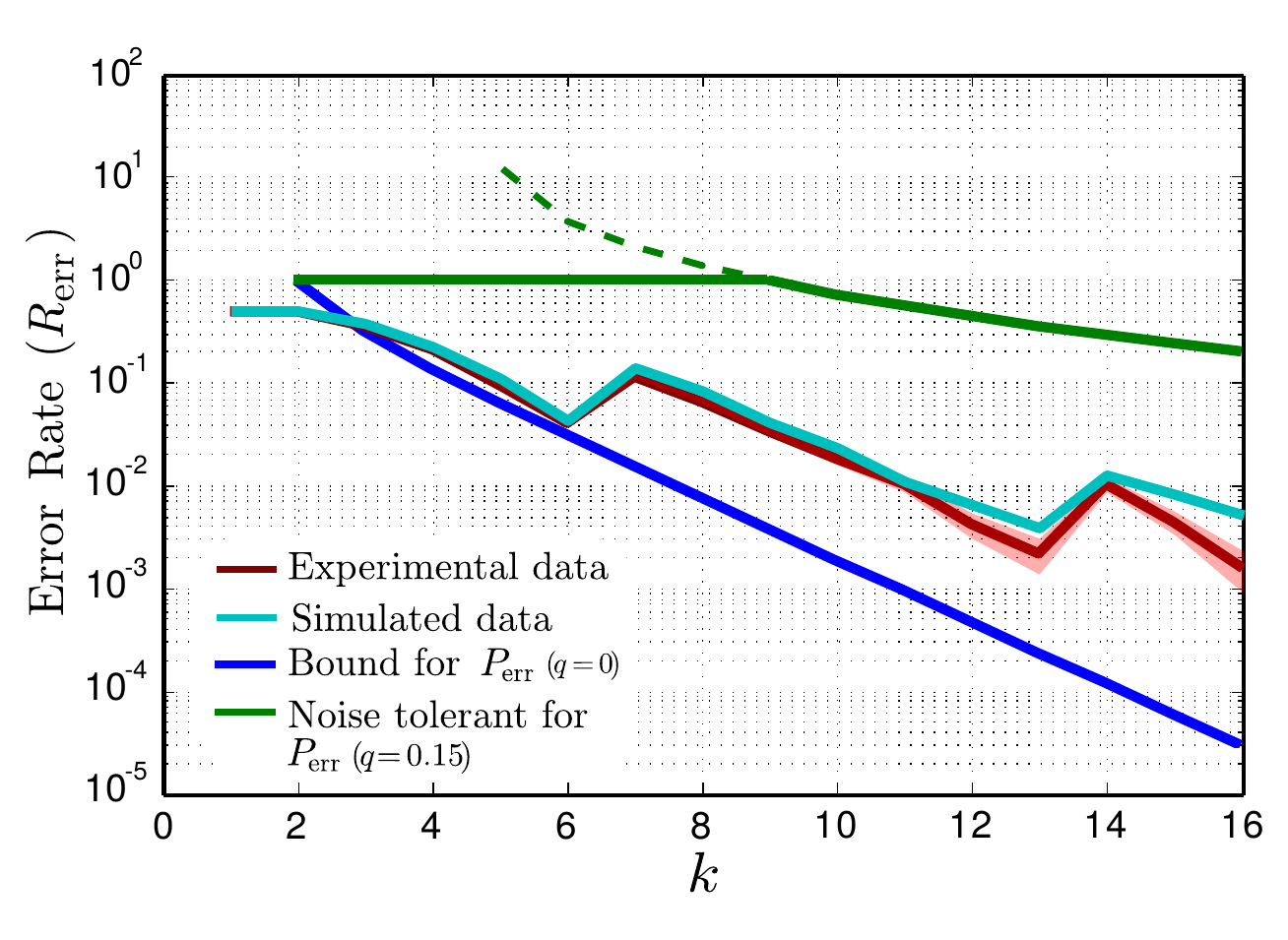}
\caption{\footnotesize{\emph{The plot shows the experimental error rate obtained with the noise tolerant protocol  (red lines), compared with the theoretical bounds for: the noiseless (blue line) and noise tolerant (green line) algorithms. The cyan line is the computational simulation of the experimental data taking into account the average $QBER$.}}}
  \label{fig:exp&sim}
\end{figure}

Figure \ref{fig:length1contrib} shows the total error rate and the contribution to this quantity arising from programs of length 1. It is evident that errors occur mostly in the minimum length programs. In particular for $k>10$ all the guessing errors come from these programs (which are the first to be evaluated in the \emph{search} procedure).
This fact simplifies the description of the error occurrence just in terms of the probability of error occurrence while evaluating length 1 programs.


\begin{figure}[h!]
    \centering
    \includegraphics[width=8.75cm]{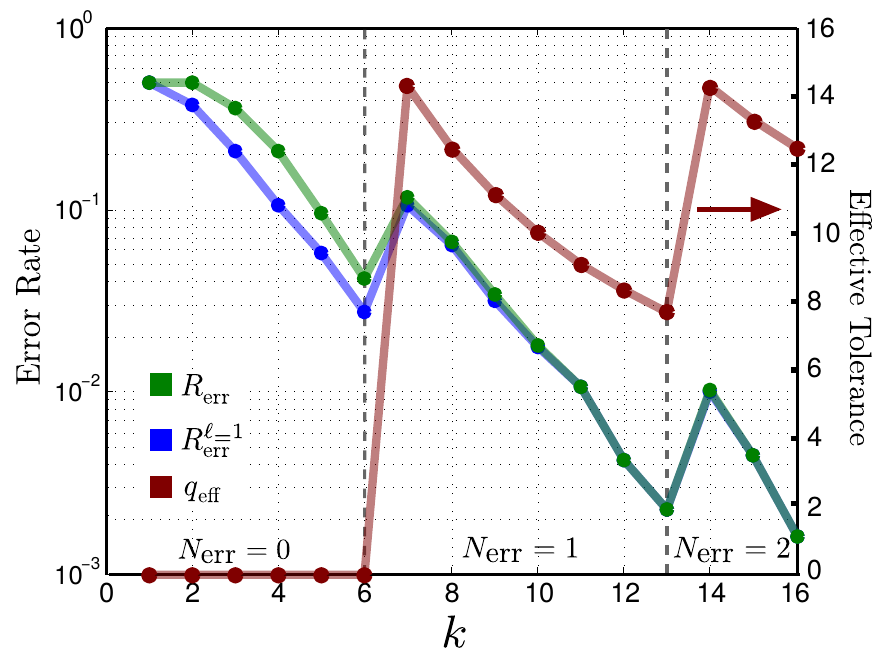}
  \caption{\footnotesize{\emph{
Total error rate (green) and the contribution to this quantity coming from the programs of length 1 (blue), together with the effective tolerance $q_{\rm{eff}}$ (dark red) for each value of $k$ over the 3100 experiment repetitions. The tolerance value $q$ was set to 0.15. Almost every error comes from $\ell=1$ programs and the probability of error occurrences coming from programs with $\ell>1$ vanishes as $k$ increases.
The sudden increases of $q_{\rm{eff}}$ in $k=7$ and $k=14$ appear due to the discrete nature of the maximum number of errors allowed on a accepted bit string (maximum Hamming distance). Also for each value of $k$ where the effective tolerance probability for $\ell=1$ increases, the error probability also increases. The vertical dashed lines delimit the regions where the maximum number of bit flips $N_{{\rm err}}$ allowed (for $\ell=1$) is constant.
}}}
\label{fig:length1contrib}
\end{figure}


The tolerance $q$ determines the maximum number of errors allowed: $N_{{\rm err}}=\lfloor q\times k \times \ell   \rfloor$. This quantity divided by the program length gives the \emph{effective tolerance}: \newline $q_{\rm{eff}}=\frac{\lfloor q\times k \times \ell  \rfloor}{\ell }$. As almost all the errors arises from minimum length programs, the $R_{\rm{err}}$ increments can be explained looking at $q_{\rm{eff}}$ from $\ell =1$.
As can be seen in figure \ref{fig:length1contrib}, for $k$ below $6$ the effective tolerance is $0$ ($N_{\rm{err}}=0$). That is why the error rate follows the ideal theoretical curve for these values (figure \ref{fig:exp&sim}). The increments on the error at $k=7$ and $k=14$ are correlated with increments in the $q_{\rm{eff}}$ (this will happen for every $k$ where the number of maximum bit flips allowed $N_{{\rm err}}$ is increased by 1 for $\ell=1$ programs).

\begin{figure}[h!]
    \centering
    \includegraphics[width=8cm]{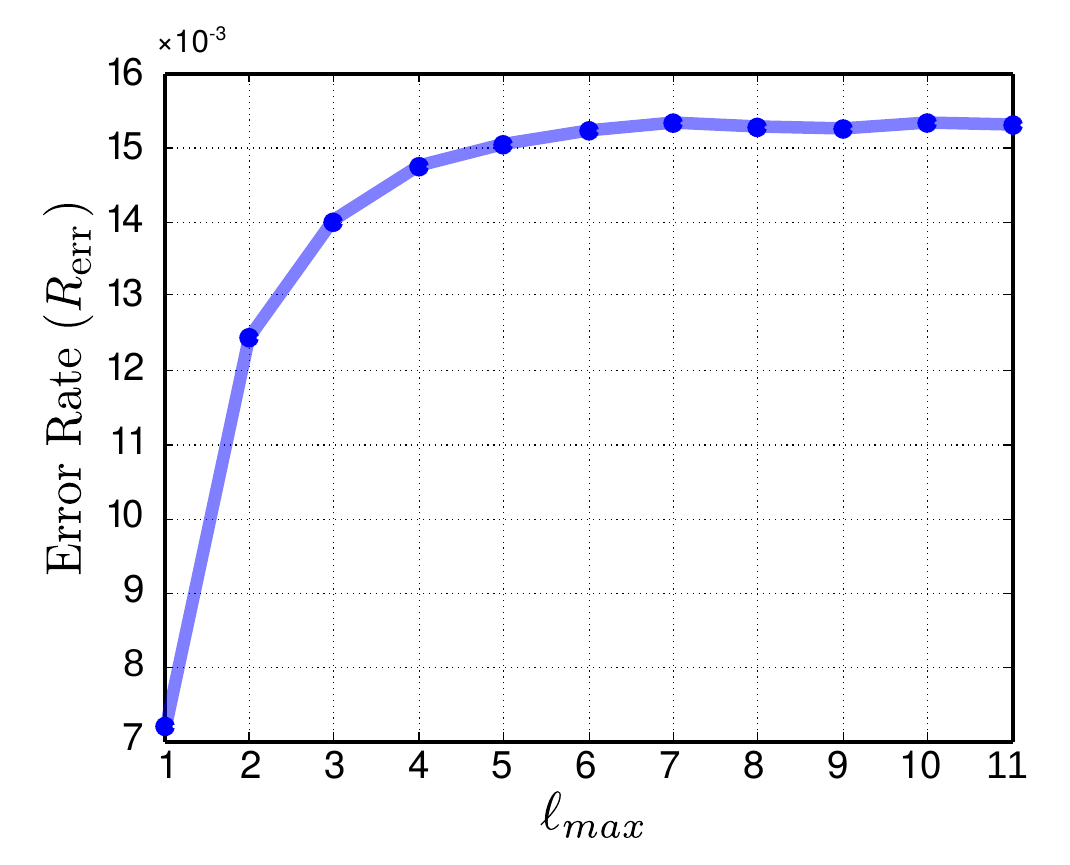}
  \caption{\footnotesize{\emph{The plotted data corresponds to simulations of the experiment with an increasing set of programs (from $2^1$ to $2^{11}$ programs). For each value of $\ell_{max}$ we run a simulation of the experiment with $5 \times 10^4$ repetitions. The $P_{{\rm err}}$ value stabilizes as the number of programs grows.
  }}}
  \label{fig:PerrLmax}
\end{figure}

Finally, as a validation of the results, we simulated the same experiment with different sizes of the set of programs used for fixed values of $k=14$ and $q=0.15$ (recall that for this experiment the program list was restricted to $2^{10}$ different elements).
Figure \ref{fig:PerrLmax} shows the simulated error rate for values of the maximum length program $\ell_{max}$ from $1$ to $11$ over $5\times10^4$ repetitions.
The error rate stabilizes as $\ell_{max}$ grows. This shows that our results are representative of the values that would be obtained if an experiment with larger $\ell_{max}$ was performed.
In this regard, a similar experiment was implemented afterwards, utilizing a larger set of seeds for the \textit{rand()} function:
The complete protocol was implemented with $\ell_{max}=16$ (65536 different seeds) over $3200$ repetitions of the experiment with $q=0.15$ and $k$ taking values from $1$ to $16$ showing the same behavior of the basis guess success rate.

    \section{Discussion}

    In this article we extended results from \cite{Bendersky2016} in two ways. First, we proved that any attempt to mix pure states into the maximally mixed state, when performed by a computer (or any system equivalent in terms of computability power), can be distinguished from the maximally mixed state prepared correctly (either the one obtained by looking at a part of a maximally entangled state or by using a truly random source). This broadens the scope of the previous results, in which only some computable mixtures were analyzed.

    Second, we presented a proof-of-concept experiment showing that mixing two different sets of pure states that are supposed to yield the same mixed state, can be distinguished when mixed employing one of the most widely used general purpose pseudorandom number generators.

These two results should be seen as a call for attention when performing experiments and claiming to produce certain mixed states via computable mixings.

    \begin{acknowledgments}
This work is supported by the Argentinian ANPCyT (PICT-2011-0365), the Laboratoire International Associée INFINIS, the ERC CoG QITBOX, the AXA Chair in Quantum Information Science, the Spanish MINECO (QIBEQI FIS2016-80773-P and Severo Ochoa SEV-2015-0522), Generalitat de Catalunya (CERCA Programme) and Fundaci\'{o} Privada Cellex.
The authors would also like to thank Laura Knoll and Christian Schmiegelow for fruitful discussions.
    \end{acknowledgments}

\appendix

\section{Appendix 1: A POVM for the generalized distinguishing protocol}

For completeness, in this section we describe an informationally complete POVM $\{E_i\}_{i\leq N_d}$ satisfying
\eqref{eqn:unbiasedness}. We construct it from the following $N_d:=d(2d-1)$ projectors
\begin{align*}
  \Pi^{(a)}_m &:= \ket{m}\bra{m},\\
  \Pi^{(b\pm)}_{n,m}&:= \frac{1}{2}[\ket{m}\bra{m}\pm\ket{m}\bra{n}\pm\ket{n}\bra{m}+\ket{m}\bra{m}],\\
  \Pi^{(c\pm)}_{n,m}&:= \frac{1}{2}[\ket{m}\bra{m}\mp i\ket{m}\bra{n}\pm i\ket{n}\bra{m}+\ket{m}\bra{m}],
\end{align*}
for all $m<n\leq d$. It is easy to see that:

\begin{enumerate}
\item $\Tr(\Pi^{(a)}_m\frac{\mathbb{I}}{d})=\Tr(\Pi^{(b\pm)}_{n,m}\frac{\mathbb{I}}{d})=\Tr(\Pi^{(c\pm)}_{n,m}\frac{\mathbb{I}}{{d}})=\frac{1}{d}$ and
\item For every density matrix $\rho$ over $\mathbb{C}^d$,
\begin{align*}
\rho_{m,m}&=\Tr(\Pi_{m,n}^{(a)}\rho),\\
\rho_{m,n}&=\frac{1}{2}\left[\Tr(\Pi^{(b+)}_{m,n}\rho)-\Tr(\Pi^{(b-)}_{m,n}\rho)\right.\\
&\left. \quad+i(\Tr(\Pi^{(c-)}_{m,n}\rho)-\Tr(\Pi^{(c+)}_{m,n}\rho)\right],\mbox{ for }m\neq n.
\end{align*}
\end{enumerate}

Finally, since $$\sum_{n,m}[\Pi^{(a)}_{n,m}+\Pi^{(b\pm)}_{n,m}+\Pi^{(c\pm)}_{n,m}]=(2d-1)\mathbb{I},$$ by normalizing these projectors with $1/(2d-1)$ we get the the effects $E_i$ of a POVM with the
desired characteristics.

\section{Appendix 2: Probability of success of Algorithm \ref{alg:distinguishNoise}}

We need to bound the
number of sequences that have a Hamming distance smaller than
$qk\ell$ from a computable one. One possible bound is $2^\ell
\binom{\ell k}{\left\lfloor q\ell k \right\rfloor} 2^{\left\lfloor q\ell k
\right\rfloor}$, where the first exponential term counts the
number of different programs of length $\ell$, the combinatorial
number corresponds to the number of bits that can be flipped due
to errors, and the last exponential term gives which of these bits
are actually being flipped. This estimation may not be tight, as
we may be counting the same sequence several times. However, using
this estimation we derive a sensible upper bound for the final
error probability, as we get

\begin{equation}
 P_{{\rm err}}<\sum_{\ell >0}\frac{2^\ell 2^{\left\lfloor q\ell k \right\rfloor}\binom{\ell k}{\left\lfloor q\ell k \right\rfloor}}{2^{\ell k}}
\end{equation}
If we consider that $q<1/2$, we can remove the integer part
function and use the generalization of combinatorial numbers for
real values. Then, by using that
$\binom{a}{b}\leq\left(\frac{ea}{b}\right)^b$, we obtain

\begin{equation}
  P_{{\rm err}}<\sum_{\ell >0}\left[2^{(1+qk-k)}  \left(\frac{e}{q}\right)^{qk}\right]^\ell.
\end{equation}
This geometric sum can be easily computed yielding

\begin{equation}
 P_{{\rm err}}<\frac{2^{1+qk-k}\left(\frac{e}{q}\right)^{qk}}{1-2^{1+qk-k}\left(\frac{e}{q}\right)^{qk}}.
\end{equation}
Now it can be numerically shown that for $q\lesssim 0.21$ the
probability of mis-recognition tends to zero exponentially with
$k$.

Finally, for completeness, we show that (with probability $1$) Algorithm \ref{alg:distinguishNoise} halts for all inputs satisfying the assumptions. Let $f<q$ be the probability of a bit flip. With probability $1$, we have that for every $\delta$ there exist an
$m_0$ such that for every $m>m_0$ the portion of bit flips in both
$X\uph m$ and $Z\uph m$ are less than $(f+\delta)m$. This
means that if we go to long enough prefixes (or programs), the
portion of bit flips will be less than $q$. And since any
computable sequence is computable by arbitrarily large programs,
this ensures that our algorithm will, at some point, come to an
end.

\end{document}